\documentclass[sigconf]{acmart}
\usepackage{citesort}
\usepackage{color}
\usepackage[ruled, vlined, linesnumbered]{algorithm2e}
\usepackage{amsmath,amssymb}
\usepackage{caption}
\DeclareCaptionType{copyrightbox}
\usepackage{framed,url}
\usepackage{tikz}
\usepackage{multirow}
\usepackage{bbm}
\usepackage{float}
\usepackage{afterpage}
\usepackage{epsfig}
\usepackage{epstopdf}
\usepackage{subfig}
\usepackage{balance}

\allowdisplaybreaks

\usepackage{graphicx}
\usepackage{subfig}
\usepackage{float}
\usepackage{afterpage}
\usepackage{color}
\usepackage{url}
\usepackage{bbm}
\usepackage{epsfig}
\usepackage{epstopdf}
\usepackage{balance}
\allowdisplaybreaks

\newtheorem{remark}{Remark}

\newcommand{\AND}{{AND}}
\newcommand{\DISJ}{{DISJ}}
\newcommand{\SUM}{{SUM}}
\newcommand{\INF}{{Gap-$l_\infty$}}

\newcommand{\eps}{\epsilon}
\newcommand{\E}{\mathbf{E}}
\newcommand{\cE}{\mathcal{E}}
\newcommand{\var}{\mathbf{Var}}
\renewcommand{\Pr}{\mathbf{Pr}}
\newcommand{\abs}[1]{\left| #1 \right|}
\newcommand{\norm}[1]{\left\lVert #1 \right\rVert}

\newcommand{\A}{\mathcal{A}}
\newcommand{\B}{\mathcal{B}}

\newcommand{\nnz}{\mathtt{nnz}}

\begin{document}

\copyrightyear{2018}
\acmYear{2018}
\setcopyright{acmlicensed}
\acmConference[PODS'18]{35th ACM SIGMOD-SIGACT-SIGAI
Symposium on Principles of Database Systems}{June 10--15,
2018}{Houston, TX, USA}
\acmBooktitle{PODS'18: 35th ACM SIGMOD-SIGACT-SIGAI Symposium
on Principles of Database Systems, June 10--15, 2018, Houston, TX,
USA}
\acmPrice{15.00}
\acmDOI{10.1145/3196959.3196964}
\acmISBN{978-1-4503-4706-8/18/06}

\fancyhead{}

\settopmatter{printacmref=false, printfolios=false}

\title{Distributed Statistical Estimation of Matrix Products \\ with Applications}

\titlenote{Qin Zhang is supported by NSF CCF-1525024 and IIS-1633215.}

\author{David P. Woodruff}
\affiliation{%
  \institution{Carnegie Mellon University}
  \city{Pittsburgh} 
  \state{PA} 
  \postcode{15213}
  \country{USA}
}
\email{dwoodruf@cs.cmu.edu}

\author{Qin Zhang}
\affiliation{%
  \institution{Indiana University Bloomington}
  \city{Bloomington} 
  \state{IN} 
  \postcode{47408}
  \country{USA}
}
\email{qzhangcs@indiana.edu}

\begin{abstract}
We consider statistical estimations of a matrix product over the integers in a distributed setting, where we have two parties Alice and Bob; Alice holds a matrix $A$ and Bob holds a matrix $B$, and they want to estimate statistics of $A \cdot B$. We focus on the well-studied $\ell_p$-norm, distinct elements ($p = 0$), $\ell_0$-sampling, and heavy hitter problems. The goal is to minimize both the communication cost and the number of rounds of communication. 

This problem is closely related to the fundamental set-intersection join problem in databases: when $p = 0$ the problem corresponds to the size of the set-intersection join. When $p = \infty$ the output is simply the pair of sets with the maximum intersection size. When $p = 1$ the problem corresponds to the size of the corresponding natural join. We also consider the heavy hitters problem which corresponds to finding the pairs of sets with intersection size above a certain threshold, and the problem of sampling an intersecting pair of sets uniformly at random. 
\end{abstract}

\maketitle

\section{Introduction}
\newcommand{\poly}{{\mathrm{poly}}}
\label{sec:intro}
We study the problem of statistical estimations of a matrix product in the distributed setting. Consider two parties Alice and Bob; Alice holds a matrix $A \in \{0,1\}^{n \times n}$ and Bob holds a matrix $B \in \{0,1\}^{n \times n}$, and they want to jointly compute a function $f$ defined on $A$ and $B$ by exchanging messages.  The goal is to minimize both the total communication cost and number of rounds of interaction.

One of the main statistical quantities we consider is the $p$-norm $\norm{C}_p$ of the product $C = A \cdot B$, defined as
$$\textstyle \norm{C}_p = \left(\sum_{i,j \in [n]} \abs{C_{i,j}}^p\right)^{{1}/{p}}.$$
Here the matrix product $A \cdot B$ is the standard matrix product over the integers.
Interpreting $0^0$ as $0$, we see that $p = 0$ corresponds to the number of non-zero
entries of $C$, which, interpreting the rows of $A$ and columns of $B$ as sets,
corresponds to the set-intersection join size (see Section~\ref{sec:app} for the formal definition).
This can also be viewed as a matrix form
of the well-studied distinct elements problem in the data stream literature
(see, e.g., \cite{FM85,BJKST02,KNW10}). 
Again interpreting the rows of $A$ and the columns of $B$ as sets,
the case $p = 1$ corresponds to the size of the corresponding natural join (again see Section~\ref{sec:app} for the formal definition).
The $p = 2$ case corresponds to the (squared) Frobenius norm of the matrix
product $A \cdot B$, which is a norm of fundamental importance in a variety
of distributed linear algebra problems, such as low rank approximation (for a
recent survey, see \cite{w14}).
The case $p = \infty$ corresponds to the pair of sets of maximum intersection
size. Estimating the largest entry in a Boolean matrix product has also been studied in the centralized setting.  We refer readers to the recent paper \cite{AR17} and references therein.

As a closely related problem, we also consider the $\ell_0$-sampling problem for which the goal is to sample each non-zero entry in $C = A B$ with probability $(1 \pm \eps) \frac{1}{\norm{C}_0}$, which corresponds to approximately outputting a random pair among the intersecting pairs of sets.  $\ell_0$-sampling is also extensively studied in the data stream literature \cite{FIS08,MW10,JST11}, and is used as a building block for sketching various dynamic graph problems (see \cite{McGregor14} for a survey).

We also study the approximate heavy hitter problem defined as follows. Let 
$$\mathrm{HH}^p_\phi(C) = \{(i,j)\ |\ C_{i,j}^p \ge \phi \norm{C}_p^p\}.
$$ 
The $\ell_p$-$(\phi, \eps)$-heavy-hitter ($0 < \eps \le \phi \le 1$) problem asks to output a set $S$ such that
$$\mathrm{HH}^p_\phi(AB) \subseteq S \subseteq \mathrm{HH}^p_{\phi - \eps}(AB).$$
As outputting the matrix product $C$ requires outputting $n^2$ numbers, it is
natural to output the set $S$ as a sparse approximation of $C$; indeed this
can be viewed as a matrix form of the well-studied compressed sensing problem.

As mentioned, these basic statistical problems, being interesting for their own sake, have strong relationships to fundamental problems in databases. We describe such relationships more formally below. 

Despite a large amount of work on computing $p$-norms and heavy hitters on frequency vectors in the streaming literature (see, e.g., \cite{Muthu05} for a survey), we are not aware of any detailed study of these basic statistical functions on matrix products.  The purpose of this paper is to introduce a systematic study of statistical estimations on matrix products.

\subsection{Motivation and Applications}
\label{sec:app}

Estimating the norm of a matrix product is closely related to two of the most important operations in relational databases -- the {\em composition} and the {\em natural join}.  Suppose we are given two relations $\A$ and $\B$, where $\A$ is defined over attributes $(X, Y)$ and $B$ is defined over attributes $(Y, Z)$. Assume for simplicity that $dom(X) = dom(Y) = dom(Z) = [n]$. We thus have $\A \subseteq [n] \times [n]$ and $\B \subseteq [n] \times [n]$.  The composition of $\A$ and $\B$ is defined to be
$$\A \circ \B = \{(i,j)\ |\ \exists k : (i, k) \in \A \wedge (k, j) \in \B\}.$$
The natural join is defined to be 
$$\A \bowtie \B = \{(i, k, j)\ |\ (i, k) \in \A \wedge (k, j) \in \B\}.$$
It is easy to see that the natural join corresponds to the composition together with the requirement that all the ``witnesses'' $k$ are output. 

We further define ``projection'' sets $A_i = \{k\ |\ (i, k) \in \A\}$ for each $i \in [n]$, and $B_j = \{k\ |\ (k, j) \in \B\}$ for each $j \in [n]$. Then we can rewrite the composition and natural joins as follows:
\begin{itemize}
\item[] $\A \circ \B = \{(i,j)\ |\ A_i \cap B_j \neq \emptyset\}$,

\item[] $\A \bowtie \B = \{(i, k, j)\ |\ k \in A_i \cap B_j\}$.
\end{itemize}
We thus also refer to compositions as {\em set-intersection joins}, and natural joins as {\em set-intersection joins with witnesses}.

As an application of set-intersection joins, consider a job application scenario: we have $n$ applicants, with the $i$-th applicant having a set of skills $A_i$ from the universe $\{1, \ldots, n\}$, and $n$ jobs, with the $j$-th job requiring a set of skills $B_j$.  Our goal is to find all the possible applicant-job matches, namely, those pairs $(i,j)$ such that $A_i \cap B_j \neq \emptyset$. One may also be interested in the number
of such matches (the $\ell_0$-norm) or the most qualified applicants (the entry 
realizing the $\ell_{\infty}$-norm, or the heavy hitters). 

We can further relate set-intersection joins to {\em Boolean matrix multiplication}. Let $A$ and $B$ be two $n \times n$ matrices such that each row $A_{i,*}$ is the indicator vector of $A_i$, and each column $B_{*,j}$ is the indicator vector of $B_j$.  Then the non-zero entries of $AB$ exactly correspond to the outputs of the set-intersection joins on $\{A_1, \ldots, A_n\}$ and $\{B_1, \ldots, B_n\}$.
If we are interested in estimates to the sizes of the joins, which are very useful for guiding query optimization since they can be computed using much less communication than computing the actual joins, then we have
\begin{itemize}
\item $\norm{AB}_0 = \abs{\A \circ \B}$, that is, the $\ell_0$-norm of $AB$ is the size of the composition of $\A$ and $\B$,

\item $\norm{AB}_1 = \abs{\A \bowtie \B}$, that is, the $\ell_1$-norm of $AB$ is the size of the natural join of $\A$ and $\B$.
\end{itemize}

Finally, $\norm{AB}_{\infty}$ corresponds to the pair $(i,j)$ with the maximum overlap, and $\{(i,j)\ |\ (AB)_{i,j} \ge \phi \norm{AB}_p\}$ for a threshold $\phi$ corresponds to the set of heavy hitters, i.e., those pairs of sets whose intersection size exceeds the threshold.
These two problems have natural applications in inner product similarity joins on a set of vectors; we refer the reader
to recent work \cite{APR016} 
on inner product similarity joins and references therein.

\begin{remark}
  We note that all of these problems and the results in this paper can be straightforwardly modified to handle the general case where
  $dom(X) = m_1$, $dom(Z) = m_2$ and $dom(Y) = n$, which corresponds to $AB$ where $A \in \{0,1\}^{m_1 \times n}$ and $B \in \{0,1\}^{n \times m_2}$.  See Section~\ref{sec:conclude} for more discussions.
\end{remark}

\subsection{Our Results}
\label{sec:result}
For simplicity we use
the notation $\tilde{O}(\cdot)$ to hide $\poly(\log \frac{n}{\eps \delta})$ factors where $\eps$ is the multiplicative approximation
ratio and $\delta$ is the error probability of a randomized communication algorithm.  We say that $X$ approximates $Y$ within a factor of
$\alpha$ if $X \in [\frac{Y}{\beta}, \gamma Y]$ where $\beta, \gamma \ge 1$ and $\beta \gamma \le \alpha$.

\medskip

\noindent{\bf Set-Intersection Join Size.}
We give a $2$-round $\tilde{O}(n/\epsilon)$-bit algorithm that approximates
$\norm{AB}_p$, $p \in [0,2]$, within a $(1+\epsilon)$ factor.  For
  the important case of
  $p = 0$, this provides a significant improvement over the
  previous $\tilde{O}(n/\eps^2)$ result in \cite{GWWZ15}. Also, due
  to the $\Omega(n/\eps^2)$ lower bound in \cite{GWWZ15} for one-round
  algorithms (i.e., algorithms for which Alice sends a single message to
  Bob, who outputs the answer), this gives a separation in the complexity
  of this problem for one and two-round algorithms. As the
  algorithm in \cite{GWWZ15} is a direct application of an
  $\tilde{O}(1/\eps^2)$ space streaming algorithm,
  our algorithm illustrates the power
  to go beyond streaming algorithms in this framework. 

\medskip

\noindent{\bf Pair of Sets with Maximum Intersection Size.}
  We first give a constant round $\tilde{O}(n^{1.5}/\eps)$-bit algorithm that approximates $\norm{AB}_{\infty}$ within a $(2+\eps)$ factor.
  We complement our algorithm by showing a few different lower bounds that hold for algorithms with any
  (not necessarily constant) number of rounds. First, we show that any algorithm that approximates $\norm{AB}_{\infty}$
  within a factor of $2$ needs $\Omega(n^2)$ bits of communication, thus necessitating our $(2+\eps)$ factor approximation.
  Moreover, we show that any algorithm achieving any constant factor
  approximation must use $\tilde{\Omega}(n^{1.5})$ bits of communication, which shows that
  our $(2+\eps)$ factor approximation algorithm has optimal communication, up to polylogarithmic
  factors.

  We next look at approximation algorithms that achieve approximation factors to $\norm{AB}_{\infty}$
  that are larger than constant. We show
  it is possible to achieve a $\kappa$-approximation factor using $\tilde{O}(n^{1.5}/\kappa)$ bits of communication.
  We complement this with an $\Omega(n^{1.5}/\kappa)$ bit lower bound. 
  
  Finally we show that the fact that the matrices $A$ and $B$ are binary is crucial. Namely, we first show
  that for general matrices $A$ and $B$ with $\poly(n)$-bounded integer entries, there is
  an $\Omega(n^2)$ lower bound for any constant factor approximation. For general approximation factors $\kappa$ 
  that may be larger than constant, we show an upper and lower bound of $\tilde{\Theta}(n^2/\kappa^2)$ communication.
  This shows an arguably surprising difference in approximation factor versus communication for binary and non-binary matrices. 

 \medskip

\noindent{\bf Heavy Hitters.}
We give an $O(1)$-round protocol that computes $\ell_p$-$(\phi, \eps)$-heavy-hitters,
$0 < \eps \le \phi \le 1$, and $p \in (0,2]$, with various tradeoffs depending
on whether Alice and Bob's matrices are arbitrary integer matrices, or whether
they correspond to binary matrices. For arbitrary integer matrices,
we achieve $\tilde{O}(\frac{\sqrt{\phi}}{\epsilon} n)$ bits of communication
for every $p \in (0,2]$. 

  We are able to significantly improve these bounds for binary matrices,
  which as mentioned above, have important applications to database joins.
  Here we show for every $p \in (0,2]$ an $O(1)$-round protocol
    with $\tilde{O}(n+ \frac{\phi}{\epsilon^2})$ bits of communication.

\subsection{Related Work}
\label{sec:related}
Early work on studying joins in a distributed model can be found in \cite{ME92} (Section $5$) and \cite{Kossmann00}. Here the goal is to output the actual join rather than its size, and such algorithms, in the worst case, do not achieve communication better than the trivial algorithm in which
Alice sends her entire input to Bob for a centralized computation.  

With the rise of the MapReduce-type models of computation, a number of works have been devoted to studying parallel and distributed computations of joins. Such works have looked at natural joins, multi-way joins, and similarity joins, in a model called the {\em massively parallel computation} model (MPC) \cite{AU11,KS11,BKS13,BKS14,KBS16,KS17,HTY17}.  Unlike our two-party communication model, in MPC there are multiple parties/machines, and the primary goal is to understand the round-load (maximum message size received by any server in any round) tradeoffs of the computation.

In a recent paper \cite{GWWZ15} the authors and collaborators studied several join problems in the two-party communication model.  The studied problems include set-intersection joins, {\em set-disjointness joins}, {\em set-equality joins}, and {\em at-least-$T$ joins}. Our results can be viewed
as a significant extension to the results in \cite{GWWZ15}, as well as a systematic study
of classical data stream problems in the context of matrix products. In particular, \cite{GWWZ15}
did not study estimating the $p$-norms of $AB$, for any $p$ other than $p = 0$. For $p = 0$,
they obtain an algorithm using $\tilde{O}(n/\eps^2)$ communication, which we significantly
improve to $\tilde{O}(n/\eps)$ communication, and extend to any $0 \leq p \leq 2$. Moreover,
we obtain the first bounds for approximating $\|AB\|_{\infty}$, where perhaps surprisingly, we
are able to obtain an $O(1)$-approximation in $\tilde{O}(n^{3/2})$ communication,
beating the na\"ive $n^2$ amount of communication. This leads us to the first algorithms
for finding the frequent entries, or heavy hitters of $AB$. 

While a number of recent works \cite{kvw14,lbkw14,blswx16,bwz16,wz16} look
at distributed linear algebra problems (for a survey, see \cite{w14}), in all
papers that we are aware of, the matrix $C$ is distributed {\it additively}. What this means
is that we want to estimate statistics of a matrix $C = A + B$, where $A$ and $B$ are held
by Alice and Bob, respectively, who exchange messages with each other. In this paper, we instead study
the setting for which we want to estimate statistics of a matrix $C = A \cdot B$, where
$A$ and $B$ are again held by Alice and Bob, respectively, who exchange messages with each other. Thus,
in our setting the underlying matrix $C$ of interest is distributed {\it multiplicatively}. 
When $C$ is distributed additively, a
common technique is for the players to agree on a random linear sketching matrix $S$, and apply it
to their inputs to reduce their size. For example, if Alice has matrix $A$ and Bob has
matrix $B$, then Alice can send $S \cdot A$ to Bob, who can compute $S(A+B)$.
A natural extension of
it in the multiplicative case is for Alice to send $S \cdot A$ to Bob, who can compute
$S \cdot A \cdot B$. This is precisely how the algorithm for $p = 0$ of \cite{GWWZ15} proceeds.
We show by using the product structure of $A \cdot B$ and more than one round,
it is possible to obtain significantly less expensive algorithms than this
direct sketching approach. 

Finally, we would like to mention several papers considering similar problems but working in the centralized model.   In \cite{c98}, Cohen uses exponential random variables and applies a minimum operation to 
obtain an unbiased estimator of the number of non-zero entries in each column of a matrix product
$C = AB$.  However, a direct adaptation of this algorithm to the distributed model would result $\tilde{\Omega}(n/\eps^2)$ bits of communication and $1$-round, which is the same as using the $1$-round $\ell_0$-sketching protocol applied to each of the columns in earlier work~\cite{GWWZ15}.  In contrast we show that 
surprisingly, at least to the authors, 
$\tilde{O}(n/\epsilon)$ bits of communication is possible with only $2$ rounds.  In \cite{ACP14}, Amossen, Campagna, and Pagh improve the time complexity of 
\cite{c98}, provided $\epsilon$ is 
not too small.  However, a direct adaptation of this algorithm to the distributed model would result an even higher communication cost of $\Omega(n^2)$.

In \cite{cl99}, the $\ell_1$-sampling problem is considered. In this paper we do not emphasize estimation of $\|C\|_1$, since this quantity can be computed exactly using $O(n \log n)$ bits of communication, as stated in Remark~\ref{rem:ell-1}.  Similarly  $\ell_1$-sampling can also be done in $O(n \log n)$ bits of communication, as illustrated in Remark~\ref{rem:ell-1-sample}.

In \cite{p13}, it is shown how to apply CountSketch to the entries of a matrix
 product $C = AB$ where $A, B \in \mathbb{R}^{n \times n}$. The time
complexity is $O(\nnz(A) + \nnz(B) + n \cdot k \log k)$, where $\nnz(A)$ denotes the number
of non-zero entries of $A$,  and $k$ is the number of hash buckets in CountSketch which is at least $1/\eps^2$. This outperforms the na\"ive time complexity of first computing
$C$ and then hashing the entries of $C$ one-by-one. While interesting from a time complexity
perspective, it does not provide an advantage over CountSketch in a distributed setting. Indeed,
for each of the hashes on Alice's side of the $n$ outer products computed in \cite{p13}, the size of
the hash is $\tilde{\Theta}(1/\epsilon^2)$, and consequently communicating this 
to Bob takes $\tilde{\Theta}(n/\epsilon^2)$ bits in total.

\section{Preliminaries}
\label{sec:preliminary}

In this section we give background on several sketching algorithms that we will make use of, as well as some basic concepts in communication complexity. We will also describe some mathematical tools and previous results that will be used in the paper.

For convenience we use $A \in \mathbb{Z}^{n \times n}$ to differentiate $A$ from a binary matrix, but we will assume that all the input matrices have polynomially bounded integer entries.  For all sketching matrices we will make use of, without explicitly stated, each of their entries can be stored in $\tilde{O}(1)$ bits.

\medskip

\noindent{\bf Sketches.\ \ }  
A sketch $sk(x)$ of a data object $x$ is a summary of $x$ of small size (sublinear or even polylogarithmic in the size of $x$) such that if we want to perform a query (denoted by a function $f$) on the original data object $x$, we can instead apply another function $g$ on $sk(x)$ such that $g(sk(x)) \approx f(x)$. Sketches are very useful tools in the development of space-efficient streaming algorithms and communication-efficient distributed algorithms. Many sketching algorithms have been developed in the data stream literature. In this paper we will make use of the following.

\begin{lemma}[\cite{Indyk00,KNW10}, $\ell_p$-Sketch $(0 \le p \le 2)$]
\label{lem:ell-p}
For $p \in [0, 2]$ and a data vector $x \in \mathbb{R}^n$, there is a sketch $sk(x) = S x$ where $S \in \mathbb{R}^{O\left(\frac{1}{\eps^2} \log\frac{1}{\delta}\right) \times n}$ is a random sketching matrix, and a function $g$ such that with probability $1 - \delta$, $g(sk(x))$ approximates $\norm{x}_p$ within a factor of $(1+\eps)$. 
\end{lemma}

\noindent{\bf Communication Complexity.\ \ }
We will use two-party communication complexity to prove lower bounds for the problems we study.  In the two-party communication complexity model, there are parties Alice and Bob.  Alice gets an input $x \in \mathcal{X}$, and Bob gets an input $y \in \mathcal{Y}$. They want to jointly compute a function $f : \mathcal{X} \times \mathcal{Y} \to \mathcal{Z}$ via a communication protocol.  Let $\Pi$ be a (randomized) communication protocol, and let $r_A, r_B$ be the private randomness used by Alice and Bob, respectively.  Let $\Pi_{X, Y, r_A, r_B}$ denote the transcript (the concatenation of all messages) when Alice and Bob run $\Pi$ on input $(X, Y)$ using private randomness $(r_A, r_B)$, and let $\Pi(X, Y, r_A, r_B)$ denote the output of the protocol. We say $\Pi$ errs with probability $\delta$ if for all $(x, y) \in \mathcal{X} \times \mathcal{Y}$, 
$$
\Pr_{r_A, r_B}[\Pi_{X, Y, r_A, r_B} \neq f(x,y)] \le \delta.
$$
We define the randomized communication complexity of $f$, denoted by $R_\delta(f)$, to be $\min_\Pi\max_{x,y,r_A,r_B} \abs{\Pi_{X, Y, r_A, r_B}}$, where $\abs{z}$ denotes the length of the transcript $z$.

We next introduce a concept called the distributional communication complexity.  Let $\mu$ be a distribution over the inputs $(X, Y)$.  We say a deterministic protocol $\Pi$ computes $f$ with error probability $\delta$ on $\mu$ if 
$$
\Pr_{(X, Y) \sim \mu}[\Pi_{X, Y} \neq f(x,y)] \le \delta.
$$
The $\delta$-error distributional communication complexity under input distribution $\mu$, denoted by $D_\delta^\mu(f)$, is the minimum communication complexity of a deterministic protocol that computes $f$ with error probability $\delta$ on $\mu$.  The following lemma connects  distributional communication complexity with randomized communication complexity.

\begin{lemma}[Yao's Lemma]
\label{lem:Yao}
For any function $f$ and any $\delta > 0$,  $R_\delta(f) \ge \max_{\mu} D_\delta^\mu(f)$.
\end{lemma}

A standard method to obtain randomized communication complexity lower bounds is to first find a hard input distribution $\mu$ for a function $f$, and then try to obtain a lower bound on the distributional communication complexity of $f$ under inputs $(X, Y) \sim \mu$. By Yao's Lemma, this is also a lower bound on the randomized communication complexity of $f$.

We now introduce two well-studied problems in communication complexity. 

\medskip

\noindent{\bf  Set-Disjointness (\DISJ).\ \ }  In this problem we have Alice and Bob. Alice holds $x = (x_1, \ldots, x_t) \in \{0,1\}^t$, and Bob holds $y = (y_1, \ldots, y_t) \in \{0,1\}^t$.  They want to compute 
$$\text{\DISJ}(x, y) = \vee_{i=1}^t (x_i \wedge y_i).$$

\begin{lemma}[\cite{BJKST04}]  
\label{lem:disj}
$R_{0.49}(\text{\DISJ}) \ge \Omega(n)$.
\end{lemma}

\medskip

\noindent{\bf \INF.\ \ }  In this problem Alice holds $x = (x_1, \ldots, x_t) \in [0, \kappa]^t$, and Bob holds $y = (y_1, \ldots, y_t) \in [0, \kappa]^t$, with the following promise: either $\abs{x_i - y_i} \le 1$ for all $i$; or for some $i$, $\abs{x_i - y_i} \ge \kappa$. Define $\text{\INF}(x, y) = 1$ if $\norm{x - y}_\infty \ge \kappa$, and $\text{\INF}(x, y) = 0$ otherwise.

\begin{lemma}[\cite{BJKST04}]  
\label{lem:inf}
$R_{0.49}(\text{\INF}) \ge \Omega(n/\kappa^2)$.
\end{lemma}

\medskip

\noindent{\bf Tools and Previous Results.\ \ }We will make use of the following results on distributed matrix multiplication and $\ell_0$-sampling on vectors.  

\begin{lemma}[\cite{GWWZ15}, Distributed Matrix Multiplication]  
\label{lem:f0}
Suppose Alice holds a matrix $A \in \mathbb{R}^{n \times n}$, and Bob holds a matrix $B \in \mathbb{R}^{n \times n}$.
There is an algorithm for Alice and Bob to compute $C_A$ and $C_B$ such that with probability $1 - 1/n^{10}$, $C_A + C_B = A B$. The algorithm uses $\tilde{O}(n \sqrt{\norm{AB}_0})$ bits of communication and $2$ rounds.  
\end{lemma}

\begin{lemma}[\cite{JST11}, $\ell_0$-Sampling]  
\label{lem:ell-0-sampling}
For a data vector $x \in \mathbb{R}^n$, there is a sketch $sk(x) = Sx$ where $S \in \mathbb{R}^{\tilde{O}(1) \times n}$ is a random sketching matrix, and a function $g$ such that $g(sk(x))$ returns $i \in [n]$ for each coordinate $x_i > 0$ with probability ${1}/{\norm{x}_0}$.  The process fails with probability at most $1/n^{10}$.  
\end{lemma}

We will also need the standard Chernoff bound.
\begin{lemma}[Chernoff Bound]
\label{lem:Chernoff-standard}
Let $X_1, \ldots, X_n$ be independent Bernoulli random variables such that $\Pr[X_i = 1] = p_i$. Let $X = \sum_{i \in [n]} X_i$. Let $\mu = \E[X]$. It holds that $\Pr[X \ge (1+\delta)\mu] \le e^{-\delta^2\mu/3}$ and $\Pr[X \le (1-\delta)\mu] \le e^{-\delta^2\mu/2}$ for any $\delta \in (0,1)$.
\end{lemma}

\section{$(1+\eps)$-Approximation of {$\ell_p\ (p \in [0, 2])$}}
\label{sec:small-p}

For notational convenience (in order to unify $\ell_0$ and $\ell_p$ for constant $p \in (0, 2]$), we define $\norm{x}_0^0 = \norm{x}_0$ to be the number of non-zero entries of $x$.  

Note that for a constant $p$, approximating $\norm{C}_p$ within a $(1+\eps)$ factor and approximating $\norm{C}_p^p$ within a $(1+\eps)$ factor are asymptotically equivalent -- we can always scale the multiplicative error $\eps$ by a factor of $p$ (a constant), which will not change the asymptotic communication complexity. We will thus use these interchangeably for convenience.

\medskip

\noindent{\bf The Idea.\ \ }  The high level idea of the algorithm is as follows.  We first perform a {\em rough} estimation -- we try to estimate the $\ell_p$-norm of each row of $C$ within a $(1+\sqrt{\eps})$ factor.  We then sample rows of $C$ with respect to their estimated ($p$-th power of their) $\ell_p$-norm, obtaining a matrix $C'$.  We finally use $C'$ to obtain a {\em finer} estimation (i.e., a $(1+\eps)$-approximation) of $\norm{C}_p^p$.

\medskip

\noindent{\bf Algorithm.\ \ }
Set parameters $\beta = \eps^{1/2}$, $\rho = 10^4 \beta^2/\eps^2 = 10^4/\eps$. The algorithm for approximating $\ell_p$-norms for $p \in [0,2]$ is  presented in Algorithm~\ref{alg:small-p}. We describe it in words below.

\begin{algorithm}[t]
\caption{$(1+\eps)$-Approximation for $\ell_p\ (p \in [0,2])$}
\label{alg:small-p}
\SetKwInOut{Input}{Input}
\SetKwInOut{Output}{Output}
\Input{Alice has a matrix $A \in \mathbb{Z}^{n \times n}$, and Bob has a matrix $B \in \mathbb{Z}^{n \times n}$. Let $C \leftarrow A B$}
\Output{A $(1+\eps)$-approximation of $\norm{C}_p^p$}
\BlankLine

 Let $S$ be the sketching matrix in Lemma~\ref{lem:ell-p}\;

Bob computes $S B^T\in \mathbb{R}^{\tilde{O}(1/\beta^2) \times n}$ of $B^T$ and sends it to Alice\;

Alice computes $\widetilde{C} \gets (S B^T A^T)^T$\;

Alice partitions the $n$ rows of $\widetilde{C}$ to (up to) $L = \log_{1+\beta} (2n^{p+1}) = O(\frac{1}{\beta} \log n)$ groups $G_1, \ldots, G_L$, such that $G_\ell$ contains all $i \in [n]$ for which $(1+\beta)^\ell \le {\norm{\widetilde{C_{i,*}}}_p^p} < (1+\beta)^{\ell+1}$\;

\ForEach{group $G_\ell\ (\ell \in [L])$}{ 
	Alice randomly samples each $i \in G_\ell$ with probability $p_\ell$, where $p_\ell = \frac{\rho}{\abs{G_\ell}} \cdot \frac{\norm{\widetilde{G_\ell}}_p^p}{\norm{\widetilde{C}}_p^p}$ where $\norm{\widetilde{G_\ell}}_p^p = \sum_{i \in G_\ell} \norm{\widetilde{C_{i,*}}}_p^p$; Alice sends $p_\ell$ to Bob\;
	Alice then replaces all non-sampled rows in $A$ with the all-$0$ vector, obtaining $A'$, and sends $A'$ to Bob\;
}

Bob computes $C' \gets A' B$, and outputs $\sum_{\ell \in [L]} \sum_{i \in G_\ell} \frac{1}{p_\ell} \norm{C'_{i,*}}_p^p$. 
\end{algorithm}

Alice and Bob first try to estimate the $\ell_p$-norm of each row in $C$ within a factor of $(1+\beta)$. This can be done by letting Bob send an $\ell_p$-sketch of $B^T$ of size $\tilde{O}(1/\beta^2)$ to Alice using the sketch in Lemma~\ref{lem:ell-p}; Alice then computes $\widetilde{C} = (S B^T A^T)^T$.  With probability $0.99$, we have that for all $i \in [n]$,
\begin{equation}
\label{eq:a-1}
\norm{\widetilde{C_{i,*}}}_p^p \in \left[\norm{C_{i,*}}_p^p, (1+\beta) \cdot \norm{C_{i,*}}_p^p\right].
\end{equation}

We note that we can set $\beta = \eps$ (instead of $\beta = \sqrt{\eps}$) and directly get a $(1+\eps)$ approximation of $\norm{C_{i,*}}_p^p$ for each row $i$ (and thus $\norm{C}_p^p$).  This is exactly what was done in \cite{GWWZ15}. However, the communication cost in this case is $\tilde{O}(n/\eps^2)$, which is higher than our goal by a factor of $1/\eps$.

\smallskip

Alice then sends Bob $\norm{\widetilde{C_{i,*}}}_p^p$ for all $i \in [n]$.  Both parties partition all the rows of $\widetilde{C}$ into up to $L = O(1/\beta \cdot \log n)$ groups $G_1, \ldots, G_L$, such that the $\ell$-th group $G_\ell$ contains all $i \in [n]$ for which
\begin{equation}
\label{eq:a-2}
(1+\beta)^\ell \le \norm{\widetilde{C_{i,*}}}_p^p < (1+\beta)^{\ell+1}.
\end{equation}
By (\ref{eq:a-1}) and (\ref{eq:a-2}), we have that for each $i \in G_\ell$,
\begin{equation}
\label{eq:a-3}
(1+\beta)^\ell \le \norm{C_{i,*}}_p^p < (1+3\beta) \cdot (1+\beta)^{\ell}.
\end{equation}

For a fixed group $G_\ell$, let $\norm{{G_\ell}}_p^p = \sum_{i \in G_\ell} \norm{{C_{i,*}}}_p^p$ and $\norm{\widetilde{G_\ell}}_p^p = \sum_{i \in G_\ell} \norm{\widetilde{C_{i,*}}}_p^p$. 
For each $\ell \in [L]$, set 
$$ p_\ell = \frac{\rho}{\abs{G_\ell}} \cdot {\norm{\widetilde{G_\ell}}_p^p}\left/{\norm{\widetilde{C}}_p^p}\right..$$  By (\ref{eq:a-1}) we have  
\begin{equation}
\label{eq:a-4}
p_\ell \in \left[\frac{1}{2} \cdot \frac{\rho}{\abs{G_\ell}} \cdot \frac{\norm{G_\ell}_p^p}{\norm{C}_p^p}, \ 2 \cdot \frac{\rho}{\abs{G_\ell}} \cdot \frac{\norm{G_\ell}_p^p}{\norm{C}_p^p}  \right]
\end{equation}

For each $\ell \in [L]$, Alice randomly samples each $i \in G_\ell$ with probability $p_\ell$. Alice then sends Bob $A'$ which consists of all the sampled rows of $A$ with other rows being replaced by all-$0$ vectors.  Bob then computes $C' = A' B$, and outputs $\sum_{\ell \in [L]} \sum_{i \in G_\ell} \frac{1}{p_\ell} \norm{C'_{i,*}}_p^p$ as the approximation to $\norm{C}_p^p$.
\medskip

We can show the following regarding Algorithm~\ref{alg:small-p}.  
\begin{theorem}
\label{thm:small-p}
For any $p \in [0,2]$, there is an algorithm that approximates $\norm{AB}_p$ for $A, B \in \mathbb{Z}^{n \times n}$ within a $(1+\eps)$ factor with probability $1 - 1/n^{10}$, using $\tilde{O}(n/\eps)$ bits of communication and $2$ rounds. 
\end{theorem}


\noindent{\bf Correctness.\ \ }
For each $\ell \in [L]$, and each $i \in G_\ell$, let $X_i^\ell$ be a $0/1$ random variable such that $X_i^\ell = 1$ if $i \in G_\ell$ is sampled by Alice, and $X_i^\ell = 0$ otherwise.
Define
$$
Z^\ell = \frac{1}{p_\ell} \sum_{i \in G_\ell} \left(\norm{C_{i,*}}_p^p -  \frac{\norm{G_\ell}_p^p}{\abs{G_\ell}} \right) X_i^\ell.
$$
It is clear that $\E[Z^\ell] = 0$.  We now compute its variance.
\begin{eqnarray*}
\var[Z^\ell] &=&  \frac{1}{p_\ell^2} \sum_{i \in G_\ell} \left(\left(\norm{C_{i,*}}_p^p -  \frac{\norm{G_\ell}_p^p}{\abs{G_\ell}} \right)^2  \var[X_i^\ell] \right) \nonumber \\
&\le& \frac{1}{p_\ell} \sum_{i \in G_\ell} \left(\norm{C_{i,*}}_p^p -  \frac{\norm{G_\ell}_p^p}{\abs{G_\ell}} \right)^2  \nonumber \\
&\le& \frac{1}{p_\ell} \sum_{i \in G_\ell} \left(3\beta \cdot \frac{\norm{G_\ell}_p^p}{\abs{G_\ell}} \right)^2 \quad \text{(by (\ref{eq:a-3}))} \nonumber \\
&=& \frac{9\beta^2 \cdot (\norm{G_\ell}_p^p)^2}{p_\ell \abs{G_\ell}} \label{eq:b-1} \\
&\le&  \frac{18 \beta^2}{\rho} \cdot \norm{G_{\ell}}_p^p \cdot \norm{C}_p^p.  \quad (\text{by (\ref{eq:a-4})})
\end{eqnarray*}

Define $Z = \sum_{\ell \in [L]} Z^\ell$. We then have $\E[Z] = 0$, and 
\begin{eqnarray*}
\var[Z] &=& \sum_{\ell \in [L]} \var[Z^\ell] \\
&\le& \frac{18\beta^2}{\rho} \cdot \norm{C}_p^p \cdot  \sum_{\ell \in [L]}\norm{G_{\ell}}_p^p \\
&\le& \frac{18 \beta^2}{\rho} (\norm{C}_p^p)^2.
\end{eqnarray*}

By Chebyshev's inequality, we have
\begin{equation*}
\Pr[\abs{Z} \ge \eps \cdot  \norm{C}_p^p] \le \frac{\var[Z]}{(\eps \cdot  \norm{C}_p^p)^2 } = \frac{18 \beta^2}{\rho \eps^2} \le 0.01.
\end{equation*}
We thus have $\abs{\sum_{\ell \in [L]} \sum_{i \in G_\ell} \frac{1}{p_\ell} \norm{C'_{i,*}}_p^p - \norm{C}_p^p} \le \eps \norm{C}_p^p$ with probability $0.99$ (conditioned on (\ref{eq:a-1}) holding, which happens with probability $0.99$ as well).

Finally note that we can always boost the success probability of the algorithm from $0.9$ to $(1 - 1/n^{10})$ using the standard median trick and paying another $O(\log n)$ factor in the communication cost (which will be absorbed by the $\tilde{O}(\cdot)$ notation).

\medskip

\noindent{\bf Complexity.\ \ }
The communication cost of sending the $\ell_p$-sketch in the first round is $O(n/\beta^2 \cdot \log n)$ words.  The cost of sending the sampled rows is bounded by $\sum_{\ell \in [L]} (p_\ell \abs{G_\ell} \cdot n)$.
Thus the total communication cost is bounded by
\begin{eqnarray*}
&&\sum_{\ell \in [L]} \left(p_\ell \abs{G_\ell} \cdot n\right) + \left( \frac{n}{\beta^2} \cdot \log n \right)\\
&=& \tilde{O}(n) \cdot \left(\rho + \frac{1}{\beta^2} \right) \\
&=& \tilde{O}\left({n}/{\eps}\right) \quad (\text{by our choices of $\rho$ and $\beta$}).
\end{eqnarray*}

It is clear that the whole algorithm finishes in $2$ rounds of communication.

\begin{remark}
\label{rem:ell-1}
We comment that for $p = 1$, $\norm{AB}_1$ can actually be computed {\em exactly} using $O(n \log n)$ bits of communication and $1$ round:  Alice simply sends $\norm{A_{*,j}}_1$ for each $j \in [n]$ to Bob, and then Bob computes $\sum_{j \in [n]} \left(\norm{A_{*,j}}_1 \cdot \norm{B_{j,*}}_1\right)$, which is exactly $\norm{AB}_1$.
\end{remark}

\begin{remark}
\label{rem:ell-1-sample}
We can also perform $\ell_1$-sampling on $C = AB$ using $O(n \log n)$ bits of communication and $1$ round.  Alice sends for each $j \in [n]$ the value $\norm{A_{*,j}}_1$ and a random sample from column $A_{*,j}$.  Bob computes for each $j \in [n]$ the value $\norm{A_{*,j}}_1 \cdot \norm{B_{j,*}}_1$ as well as $\sum_{j \in [n]} \left(\norm{A_{*,j}}_1 \cdot \norm{B_{j,*}}_1\right)$, from which he samples a $j \in [n]$ proportional to $\norm{A_{*,j}}_1 \cdot \norm{B_{j,*}}_1$.  Finally, Bob samples a random entry $b \in B_{j,*}$, and if $a \in A_{*,j}$ is the uniform sample in $A_{*,j}$ that Alice sent to Bob, Bob outputs the pair $(a, b)$ as the $\ell_1$-sample.
\end{remark}

\subsection{$\ell_0$-Sampling}
\label{sec:sampling}

We now present a simple algorithm for $\ell_0$-sampling. Recall that the goal of $\ell_0$-sampling on matrix $C = AB$ is to sample each non-zero entry in $C$ with probability $(1 \pm \eps) \frac{1}{\norm{C}_0}$.  

The idea is fairly simple: we employ an $\ell_0$-sketch and $\ell_0$-samplers in parallel. We first use the $\ell_0$-sketch to sample a column of $C$ proportional to its $\ell_0$-norm, and then apply the $\ell_0$-sampler to that column.  For the first step, we use the one-way $\ell_0$-sketching algorithm in Lemma~\ref{lem:ell-p} to approximate the $\ell_0$-norm of each column of $C$ within a factor of $1 + \eps$. For the second step, we use the one-way $\ell_0$-sampling algorithm for vectors in Lemma~\ref{lem:ell-0-sampling} for each column of $C$.

\begin{theorem}
\label{thm:sampling}

There is an algorithm that performs $\ell_0$-sampling on $C$ with success probability $0.9$ using $\tilde{O}(n/\eps^2)$ bits of communication and $1$ round.
\end{theorem}

\begin{proof}
The size of the $\ell_0$-sampler (i.e., the sketching matrix $S$) in Lemma~\ref{lem:ell-0-sampling} is bounded by $\tilde{O}(n)$, and the size of the $\ell_0$-sketch in Lemma~\ref{lem:ell-p} is bounded by $\tilde{O}(n/\eps^2)$.  Thus the total number of bits of communication is bounded by $\tilde{O}(n/\eps^2) + \tilde{O}(n) = \tilde{O}(n/\eps^2)$.  The algorithm finishes in $1$ round since both the $\ell_0$-sketch and $\ell_0$-sampler can be computed in one round.

The success probability follows from a union bound on the success probabilities of the $\ell_0$-sketch and $\ell_0$-sampler for each of the $n$ columns of $C$.
\end{proof}

\section{$(2+\eps)$-Approximation of $\ell_\infty$}
In this section we give almost tight upper and lower bounds for approximating $\norm{C}_\infty$, that is, the maximum entry in the matrix product $C$.  We first consider the product of binary matrices, and then consider the product of general matrices.

\subsection{Upper Bounds for Binary Matrices}
\label{sec:up-infty}

\subsubsection{An Upper Bound for $2+\eps$ Approximation} \hfill
\label{sec:up-infty-1}
\smallskip

\noindent{\bf The Idea.\ \ }  The high level idea is to scale down each entry of $C$ so that $\norm{C}_1$ is as small as possible subject to the constraint that the largest entry of $C$ is still approximately preserved (after scaling back).  This down-scaling can be done by sampling each $1$-entry of $A$ with a certain probability (we replace the non-sampled $1$'s by $0$'s). Let $A'$ be the matrix of $A$ after applying sampling.  Alice and Bob then communicate for each item $j \in [n]$ the number of rows and columns in $A'$ and $B$ respectively that contain item $j$ (i.e., those rows and columns with $j$-th coordinate equal to $1$), and the one with the smaller number sends all the indices of those rows/columns to the other party.  After this, Alice and Bob can compute matrices $C_1$ and $C_2$ independently such that $C \approx C_1 + C_2$, and then output $\max\{\norm{C_1}_\infty, \norm{C_2}_\infty\}$ as an approximation to $\norm{C}_\infty$.

\medskip

\noindent{\bf Algorithm.\ \ }  
Let $L = \log_{1+\eps} \norm{A}_1 = O(\frac{\log n}{\eps})$.  Set $\gamma = \frac{10^4 \log n}{\eps^2}$.
We present the algorithm in Algorithm~\ref{alg:infty}, and describe it in words below.

\begin{algorithm}[t]
\caption{$(2+\eps)$-Approximation for $\ell_\infty$}
\label{alg:infty}
\SetKwInOut{Input}{Input}
\SetKwInOut{Output}{Output}
\Input{Alice has a matrix $A \in \{0,1\}^{n \times n}$, and Bob has a matrix $B \in \{0,1\}^{n \times n}$. Let $C \leftarrow A B$}
\Output{A $(2+\eps)$-approximation of $\norm{C}_\infty$}
\BlankLine

\ForEach{$\ell \gets 0, 1, \ldots, L$}{
	Alice samples each `1' in $A$ with probability $p_\ell = 1/(1+\eps)^\ell$ (and replaces those non-sampled 1's by 0's), obtaining matrix $A^\ell$\;  
	Let $C^\ell \gets A^\ell B$\;
}

\ForEach{$\ell \gets 0, 1, \ldots, L$}{
	Alice and Bob compute $\norm{C^\ell}_1$ using Remark~\ref{rem:ell-1}\;  
	Let $\ell^*$ be the smallest index $\ell \in \{0, 1, \ldots, L\}$ for which $\norm{C^\ell}_1 \le \gamma n^2$\; \label{line:b-1}
}

\ForEach{$j \in [n]$\label{line:b-2}}{
	Let $u_j \gets \abs{\{i \in [n]\ |\ j \in A^{\ell^*}_i\}}$, and $v_j \gets \abs{\{i \in [n]\ |\ j \in B_i\}}$\;
	\If{$u_j \le v_j$}{
		Alice sends $I_j \gets \{i \ |\ j \in A^{\ell^*}_i\}$ to Bob\;
	}
	\Else{Bob sends $I_j \gets \{i \ |\ j \in B_i\}$ to Alice\;}  \label{line:b-3}
}

Alice and Bob use the $I_j$'s to compute matrices $C_A$ and $C_B$ respectively such that \ $C^{\ell^*} = C_A + C_B$\;

Alice and Bob compute $\norm{C_A}_\infty$ and $\norm{C_B}_\infty$, and output $\max\{\norm{C_A}_\infty / p_{\ell^*}, \norm{C_B}_\infty / p_{\ell^*}\}$.
\end{algorithm}

For $\ell = 0, 1, \ldots, L$, Alice samples each $1$-entry in $A$ with probability $p_\ell = 1/(1+\eps)^\ell$ (i.e., with probability $(1 - p_\ell)$ the $1$-entry is replaced by a $0$-entry). Let $A^\ell$ be the matrix after sampling $A$ with probability $p_\ell$, and let $C^\ell = A^\ell B$.

For each $\ell = 0, 1, \ldots, L$, Alice and Bob compute $\norm{C^\ell}_1$ using Remark~\ref{rem:ell-1}.  Let $\ell^*$ be the smallest index  $\ell \in \{0, 1, \ldots, L\}$ such that $\norm{C^\ell}_1 \le \gamma n^2$.

Let us focus on $A^{\ell^*}$ and $B$, and consider each item $j \in [n]$. For convenience we identify the rows of $A^{\ell^*}$ and columns of $B$ as sets $\{A^{\ell^*}_{1}, \ldots, A^{\ell^*}_{n}\}$ and $\{B_{1}, \ldots, B_{n}\}$ respectively.
Suppose $j$ appears $u_j$ times in Alice's sets, and $v_j$ times in Bob's sets.  Alice and Bob exchange the information of $u_j$ and $v_j$ for all $j \in [n]$.  Then for each $j \in [n]$, if $u_j \le v_j$ then Alice sends all the indices of sets $A^{\ell^*}_i$ containing $j$ to Bob, otherwise Bob sends all the indices of sets $B_i$ containing $j$ to Alice.

At this point, Alice and Bob can form matrices $C_A$ and $C_B$ respectively so that $C_A+C_B = C^{\ell^*}$, where $C_A$ corresponds to the portion of each entry of $C^{\ell^*}$ restricted to the items $j$ for which Alice knows the intersections (in other words, Alice knows the inner product defining the entry $C^{\ell^*}$ restricted to a certain subset of items), and similarly define $C_B$.  Finally Alice and Bob output $\max\{\norm{C_A}_\infty / p_{\ell^*}, \norm{C_B}_\infty / p_{\ell^*}\}$ as the approximation of $\norm{C}_\infty$.

We have the following theorem.
\begin{theorem}
\label{thm:up-infty}
Algorithm~\ref{alg:infty} approximates $\norm{AB}_\infty$ for two Boolean matrices $A, B \in \{0,1\}^{n \times n}$ within a $(2+\eps)$ factor with probability $0.9$ using $\tilde{O}(n^{1.5}/\eps)$ bits of communication and $3$ rounds.
\end{theorem}

\medskip

\noindent{\bf Correctness.\ \ }  
 We first show that the claimed approximation holds.  The following lemma is a key ingredient.

\begin{lemma}
\label{lem:infty}
With probability $1 - 1/n^2$, $\norm{C^{\ell^*}}_\infty/p_{\ell^*}$ approximates $\norm{C}_\infty$ within a factor of $1+\eps$.
\end{lemma}

\begin{proof}
We assume that $\norm{C}_1 > \gamma n^2$ since otherwise there is nothing to prove (in this case we have $p_{\ell^*} = 1$ and $C^{\ell^*} = C$).  

We first define a few events.
\begin{enumerate}
\item[$\cE_1$:] $\norm{C^{\ell^*}}_\infty \ge \frac{1}{2} \gamma$.

\item[$\cE_2$:] For all pairs $(i,j)$, if $C^{\ell^*}_{i,j} \ge \frac{1}{8} \gamma$, then $C^{\ell^*}_{i,j}/p_{\ell^*}$ approximates $C_{i,j}$ within a factor of $1 + \eps$.

\item[$\cE_3$:] For all pairs $(i,j)$, if $C^{\ell^*}_{i,j} < \frac{1}{8} \gamma$, then $C_{i,j} < \frac{1}{4} \gamma / p_{\ell^*}$. 
\end{enumerate}
In words, $\cE_1$ states that the maximum entry of $C^{\ell^*}$ will be large.  $\cE_2$ states that for all large entries $(i,j)$ in $C^{\ell^*}$, the values $C^{\ell^*}_{i,j}$, after rescaling by a factor of $1/p_{\ell^*}$, can be used to approximate $C_{i,j}$ within a factor of $1 + \eps$.  $\cE_3$ states that for all small entries $(i,j)$ in $C^{\ell^*}$, the corresponding values $C_{i,j}$ cannot be the maximum in the matrix $C$.  

It is not difficult to see that if all three events hold then Lemma~\ref{lem:infty} holds.  Indeed, by $\cE_2$ we can approximate each $C_{i,j}$ by $C^{\ell^*}_{i,j} / p_{\ell^*}$ within a factor of $1 + \eps$ as long as $C^{\ell^*}_{i,j} \ge \frac{1}{8} \gamma$, and by $\cE_1$ we have $\norm{C^{\ell^*}}_\infty \ge \frac{1}{2} \gamma$. Therefore
\begin{equation}
\label{eq:c-1}
\norm{C}_\infty \ge \frac{1}{2} \gamma / (p_{\ell^*} (1+\eps)) > \frac{1}{4} \gamma / p_{\ell^*}.
\end{equation}
By $\cE_3$, for all $(i,j)$ with $C^{\ell^*}_{i,j} < \frac{1}{8} \gamma$, we have $C_{i,j} < \frac{1}{4} \gamma / p_{\ell^*}$; by (\ref{eq:c-1}) we know that these entries $(i,j)$ cannot be the maximum in $C$.
We can thus conclude that $\norm{C^{\ell^*}}_\infty$ approximates $\norm{C}_\infty/p_{\ell^*}$ within a factor of $1+\eps$.
\smallskip

In the rest of this section we show that each of $\cE_1, \cE_2, \cE_3$ holds with probability $1 - 1/n^4$.  The success probability in Lemma~\ref{lem:infty} follows by a union bound.
\smallskip

For $\cE_1$,  we only need to show that $\norm{C^{\ell^*}}_1 \ge \frac{1}{2} \gamma n^2 $.  Recall that $\ell^*$ is the smallest index $\ell \in \{0, 1, \ldots, L\}$ such that $\norm{C^\ell}_1 \le  \gamma n^2$.  We thus have $\norm{C^{\ell^*-1}}_1 >  \gamma n^2$.
We can view $C^{\ell^*}$ as sampling each entry of $C^{\ell^*-1}$ with probability $1/(1+\eps)$.  By a Chernoff bound, with probability $1 - 1/n^{10}$ we have $\norm{C^{\ell^*}}_1 \ge \frac{1}{2} \gamma n^2$.  
Consequently, we have $\norm{C^{\ell^*}}_\infty \ge \norm{C^{\ell^*}}_1/n^2 \ge \frac{1}{2} \gamma$.
\smallskip

For $\cE_2$, let us first focus on a particular pair $(i,j)$.  Let $z = C_{i,j}$, and let $k_1, \ldots, k_z \in [n]$ be the indices for which $A^{\ell^*}_{i, k_t} = B_{k_t, j} = 1$ for all $t = 1, \ldots, z$.  For each $t \in [z]$, define the random variable $X_t$ such that $X_t = 1$ if $A^{\ell^*}_{i, k_t}$ is sampled in $A^{\ell^*}$, and $X_t = 0$ otherwise.  Let $X = \sum_{t \in [z]} X_t$. We thus have $X = C^{\ell^*}_{i,j}$, and
\begin{equation} 
\label{eq:c-2}
\textstyle \E[X] = \sum_{t \in [z]} \E[X_t] = p_{\ell^*} \cdot z.
\end{equation}

The claim is $\E[X] \ge \frac{1}{16} \gamma$ with probability $1 - 1/n^{10}$.  Suppose to the contrary that $\E[X] < \frac{1}{16} \gamma$. We can just consider the case that $\E[X] \in [\frac{1}{32} \gamma, \frac{1}{16} \gamma)$ and argue that with probability $1 - 1/n^{10}$ we have $X < \frac{1}{8} \gamma$, which contradicts the assumption of $\cE_2$ that $X = C^{\ell^*}_{i,j} \ge \frac{1}{8} \gamma$.  Note that this is sufficient since if $\E[X] < \frac{1}{32} \gamma$ then the probability that $X < \frac{1}{8} \gamma$ will be even higher. In the case when $\E[X] \in [\frac{1}{32} \gamma, \frac{1}{16} \gamma)$, by a Chernoff bound we have 
$$X \in [(1-\eps) \E[X], (1+\eps) \E[X]] \subseteq \left[\frac{1}{64} \gamma, \frac{1}{8} \gamma \right)$$ 
with probability $1 - 1/n^{10}$. 

Now in the case that $\E[X] \ge \frac{1}{16} \gamma$,
by another Chernoff bound we have $X \in [(1-\eps) \E[X], (1+\eps) \E[X]]$ with probability $1 - 1/n^{10}$; in other words, $X / p_{\ell^*} (= C^{\ell^*}_{i,j} / p_{\ell^*})$  approximates $\E[X] / p_{\ell^*} (= z = C_{i,j})$ within a factor of $1 + \eps$.  
Finally, by a union bound on at most $n^2$ pairs $(i,j)$, the probability that $\cE_2$ holds is at least $1 - 1/n^4$.
\smallskip

For $\cE_3$, we again focus on a particular pair $(i,j)$, and will reuse the notation in the analysis of $\cE_2$. 
The observation is that if $\E[X] \ge \frac{1}{4} \gamma$, then $X \ge (1 - \eps) \E[X] \ge \frac{1}{8} \gamma$  with probability $1 - 1/n^{10}$, contradicting the assumption of $\cE_3$.  We thus have $C_{i,j} = z  = \E[X] / p_{\ell^*} < \frac{1}{4} \gamma / p_{\ell^*}$ with probability $1 - 1/n^{10}$.  Finally by a union bound on at most $n^2$ pairs of $(i,j)$,  the probability that $\cE_3$ holds is at least $1 - 1/n^4$.
\end{proof}

We now wrap up the correctness proof of the theorem. At the end of Algorithm~\ref{alg:infty} Alice and Bob obtain two matrices $C_A$ and $C_B$ such that $C_A + C_B = C^{\ell^*}$. We thus have $\max\{\norm{C_A}_\infty, \norm{C_B}_\infty\} \ge \norm{C^{\ell^*}}_\infty / 2$.  Combining this with Lemma~\ref{lem:infty} we obtain 
$$\frac{\norm{C}_\infty}{2(1+\eps)} \le \max\left\{\frac{\norm{C_A}_\infty}{p_{\ell^*}}, \frac{\norm{C_B}_\infty}{p_{\ell^*}}\right\} \le  (1 + \eps) \norm{C}_\infty.$$

\medskip

\noindent{\bf Complexity.\ \ } 
By Remark~\ref{rem:ell-1}, the step of computing $\norm{C^\ell}_1$ for all $\ell = 0, 1, \ldots, L$ costs $\tilde{O}(L \cdot n) = \tilde{O}(n)$ bits.  The exchanging of $\{u_j, v_j\ |\ j \in [n]\}$ costs $\tilde{O}(n)$ bits.  The last step of computing $\max\{\norm{C_A}_\infty, \norm{C_B}_\infty\}$ costs $\tilde{O}(1)$ bits.

Now we consider the step of exchanging the indices of sets containing $j$ for each $j \in [n]$. We analyze two cases. In the case that $u_j, v_j > \sqrt{n}/\eps$, there will be at most 
$$\frac{\norm{C^{\ell^*}}_1}{u_j \cdot v_j} \le \frac{\gamma n^2}{u_j \cdot v_j}$$ such items $j$. The total communication for such $j$'s is bounded by
\begin{eqnarray*}
\sum_{j: u_j, v_j > \sqrt{n}/\eps} \min\{u_j, v_j\} 
&\le& \sum_{\ell \ge 0} \frac{\gamma n^2}{ n / \eps^2 \cdot 2^{2\ell}} \cdot  \sqrt{n} / \eps \cdot  2^\ell \\
&=& \tilde{O}(\gamma \eps n^{1.5}) = \tilde{O}(n^{1.5}/\eps).
\end{eqnarray*}
In the case that $\min\{u_j, v_j\} \le \sqrt{n}/\eps$, we directly have 
\begin{eqnarray*}
\sum_{j: \min\{u_j, v_j\} \le \sqrt{n}/\eps} \min\{u_j, v_j\} 
&\le& \sum_{j \in [n]} \sqrt{n} / \eps \le n^{1.5} / \eps.
\end{eqnarray*}
Summing up, the total communication cost is bounded by $\tilde{O} (n^{1.5} / \eps)$.

Finally we show that Algorithm~\ref{alg:infty} can be implemented in $3$ rounds.  In Round $1$, for each level $\ell$ Alice sends Bob $\{\norm{A_{*,j}}_1\ |\ j \in [n]\}$ so that Bob can compute $\norm{AB}_1$ according to Remark~\ref{rem:ell-1}, and consequently finds $\ell^*$.  In Round $2$, Bob sends $\ell^*$ to Alice, together with all $I_j$ corresponding to those $j$ with $u_j > v_j$.  In Round $3$, Alice sends Bob all $I_j$ corresponding to those $j$ with $u_j \le v_j$. Alice also forms $C_A$, computes and sends $\norm{C_A}_\infty$ to Bob.  Finally Bob forms $C_B$, and computes $\max\{\norm{C_A}_\infty,  \norm{C_B}_\infty\}$ as the final output. 

\subsubsection{An Upper Bound for General $\kappa$-Approximation} \hfill
\label{sec:up-infty-2} 

\smallskip

\noindent{\bf The Idea and Algorithm.\ \ }
We next consider protocols obtaining a $\kappa$-approximation to $\norm{C}_\infty$ for a general approximation factor $\kappa > 1$.  One way to do this is to exactly follow Algorithm~\ref{alg:infty}. That is, we first scale down the entries of $C$ by sampling the $1$-entries in $A$ to a level for which $\norm{C^\ell}_1 \le \alpha n^2 / \kappa$ where $\kappa$ is the approximation ratio, and $\alpha = \Theta(\log n)$. If we continue to follow Algorithm~\ref{alg:infty}, then we will get an $\tilde{O}(n^{1.5}/\sqrt{\kappa})$ bound.  We now show how to improve the bound to $\tilde{O}(n^{1.5}/\kappa)$.

The main change we make to Algorithm~\ref{alg:infty} is that we add a universe sampling step at the beginning.  More precisely, we sample each column of $A$ with probability $q = \min\{\alpha / \kappa, 1\}$ where $\alpha = 10^4 \log n$, and then replace all non-sampled columns in $A$ with all-$0$ vectors, obtaining a new matrix $A'$.  Let $D = A' B$. Recall that $C = AB$. We compute $\norm{C}_1$ and $\norm{D}_1$.  

With this new universe sampling step it is possible to have $\norm{D}_1 = 0$.  If this happens then we also check $\norm{C}_1$. If $\norm{C}_1 = 0$ then we simply output $0$; otherwise we output $1$.  If $\norm{D}_1 > 0$, then we follow Algorithm~\ref{alg:infty} to do further sampling on $A'$, obtaining $A^1, A^2, \ldots$.  Let $C^\ell = A^\ell B$ for $\ell = 1, 2, \ldots$.  We again stop at the first level $\ell^*$ for which $\norm{{C^{\ell^*}}}_1 \le \alpha n^2/\kappa$, and then exchange for each (surviving) universe item $j$ the indices of sets that contain $j$, in exactly the same way as that in Algorithm~\ref{alg:infty}.

\begin{algorithm}[t]
\caption{$\kappa$-Approximation for $\ell_\infty$}
\label{alg:infty-2}
\SetKwInOut{Input}{Input}
\SetKwInOut{Output}{Output}
\Input{Alice has a matrix $A \in \{0,1\}^{n \times n}$, and Bob has a matrix $B \in \{0,1\}^{n \times n}$. Let $C \leftarrow A B$}
\Output{A $\kappa$-approximation of $\norm{C}_\infty$}
\BlankLine

Set $q = \min\{\alpha / \kappa, 1\}$ where $\alpha = 10^4 \log n$\;

Alice samples each column of $A$ with probability $q$ (and replaces those non-sampled columns by the all-0 vector), obtaining $A'$. Let $D \gets A' B$\;  

Alice and Bob compute $\norm{D}_1$ and $\norm{C}_1$\;

\If{$\norm{D}_1 = 0$}{
	\lIf{$\norm{C}_1 = 0$}{Output $0$}
	\lElse{Output $1$}
}
\Else{
	Follow Algorithm~\ref{alg:infty} and further sample $A'$ with probability $p_\ell = 1/2^\ell$ (instead of $p_\ell = 1/(1+\eps)^\ell$) for $\ell = 0, 1, \ldots, \log_2 \norm{A'}_1$, and with the threshold $\gamma n^2$ at Line~\ref{line:b-1} being replaced by $\alpha/\kappa \cdot n^2$. Finally output $\max\{\norm{C_A}_\infty / (q \cdot p_{\ell^*}), \norm{C_B}_\infty / (q \cdot p_{\ell^*})\}$.
}
\end{algorithm}

The algorithm is presented in Algorithm~\ref{alg:infty-2}. 
We have the following theorem.
\begin{theorem}
\label{thm:up-infty-2}

Algorithm~\ref{alg:infty-2} approximates $\norm{AB}_\infty$ for two Boolean matrices $A, B \in \{0,1\}^{n \times n}$ within a factor of $\kappa$ for any $\kappa \in [4, n]$ with probability $0.9$ using $\tilde{O}(n^{1.5}/\kappa)$ bits of communication and $O(1)$ rounds.
\end{theorem}

\medskip

\noindent{\bf Correctness.\ \ }
For simplicity we assume that $\alpha/\kappa \le 1$ (and thus $q = \alpha/\kappa$), since otherwise $D = C$ and the arguments will follow those in Algorithm~\ref{alg:infty}.

We define two events, and will show that each holds with probability $1 - 1/n^4$. 
\begin{enumerate}
\item[$\cE_4$:] For all pairs $(i,j)$, if $D_{i,j} \ge \frac{1}{8} \alpha$, then $D_{i,j}/q$ approximates $C_{i,j}$ within a factor of $2$.

\item[$\cE_5$:] For all pairs $(i,j)$, if $D_{i,j} < \frac{1}{8} \alpha$, then $C_{i,j} < \frac{1}{4} \alpha / q$. 
\end{enumerate}

We first assume that $\norm{D}_\infty > 0$. Consider a pair $(i,j)$, if $D_{i,j} < \frac{1}{8} \alpha$, then we know by $\cE_5$ that $C_{i,j} < \frac{1}{4} \alpha / q = \frac{1}{4} \kappa$. Otherwise if $D_{i,j} \ge \frac{1}{8} \alpha$ then by $\cE_4$ we know that $D_{i,j}/q$ approximates $C_{i,j}$ within a factor of $2$.  We thus conclude that $\norm{D}_\infty$ approximates $\norm{C}_\infty$ within a factor of $\kappa/4$ if $\norm{D}_\infty > 0$.  

In the case that $\norm{D}_\infty = 0$, by $\cE_5$ we know that all entries in $C$ are less than $\kappa/4$. Then we can test whether $\norm{C}_1 > 0$. If the answer is yes then we can output $1$, which already approximates $\norm{C}_\infty$ within a factor of $\kappa$; otherwise we know that $C$ is the zero matrix, and we can output $0$.

The proofs that each of $\cE_4$ and $\cE_5$ hold with probability $1 - 1/n^4$ are analogous to those for $\cE_2$ and $\cE_3$ in the proof of Lemma~\ref{lem:infty}.

\medskip

\noindent{\bf Complexity.\ \ }  The analysis of the communication cost is again similar to that of Algorithm~\ref{alg:infty}, and the bottleneck is still the exchange of the indices of sets containing $j$ for each $j \in [n]$.   
We again analyze two cases. Note that after sampling we have $\norm{C^{\ell^*}}_1 = \tilde{O}(n^2/\kappa)$, {\em and} the universe size is $\tilde{O}(n/\kappa)$. 
\begin{itemize}
\item If $\min\{u_j, v_j\} \le \sqrt{n}$, then since the universe size is $\tilde{O}(n/\kappa)$, the total communication is upper bounded by $\tilde{O}(n / \kappa) \cdot \sqrt{n} = \tilde{O}(n^{1.5}/\kappa)$.

\item If $\min\{u_j, v_j\} > \sqrt{n}$, then since $\norm{C^{\ell^*}}_1 = \tilde{O}(n^2/\kappa)$, the total communication is upper bounded by 

$\norm{C^{\ell^*}}_1 / \sqrt{n} = \tilde{O}(n^{1.5}/\kappa)$.
\end{itemize}
Therefore the total communication is bounded by $\tilde{O}(n^{1.5}/\kappa)$. The number of rounds is clearly bounded by $O(1)$.

\subsection{Lower Bounds for Binary Matrices}
\label{sec:lb-infty}

In this section we show that our algorithms for $\ell_\infty$-norm estimation in Section~\ref{sec:up-infty} are almost tight in the sense that (1) $\Omega(n^2)$ bits of communication is needed if we want to go beyond  a $2+\eps$ approximation, and (2) for any approximation $\kappa$ we need to use $\Omega(n^{\frac{3}{2}}/\kappa)$ bits of communication.

\subsubsection{A Lower Bound for $2$-Approximation}
\label{sec:lb-infty-1}

\begin{theorem}
\label{thm:lb-infty}
Any algorithm that approximates $\norm{AB}_\infty$ for two Boolean matrices $A, B \in \{0,1\}^{n \times n}$ within a factor of $2$ with probability $0.51$ needs $\Omega(n^2)$ bits of communication, even if we allow an unbounded number of communication rounds.
\end{theorem}

\begin{proof}
We perform a reduction from the two-player set-disjointness (see Section~\ref{sec:preliminary}) on strings of length $(n/2)^2 = n^2/4$, where Alice has $x$ and Bob has $y$. Alice creates an $n/2 \times n/2$ matrix $A'$ indexed by the coordinates in $x$, that is, the $i$-th $(i = 1, \ldots, n/2)$ row of $A'$ consists of the $((i-1)\frac{n}{2}+1)\text{-th}, \ldots, \frac{i n}{2}$-th coordinates of $x$. Similarly, Bob creates an $n/2 \times n/2$ matrix $B'$ indexed by the coordinates in $y$.  Next, Alice creates an $n \times n$ input matrix 
 \[
   A=
  \left[ {\begin{array}{cc}
   A' & I \\
   \mathbf{0} & \mathbf{0} \\
  \end{array} } \right],
\]
where $I$ is an $n/2 \times n/2$ identity matrix, and $\mathbf{0}$ is an $n/2 \times n/2$ all-$0$ matrix.
 Bob creates an $n \times n$ input matrix
\[
   B=
  \left[ {\begin{array}{cc}
   I & \mathbf{0} \\
   B' & \mathbf{0} \\
  \end{array} } \right].
\]
Note that $A$ and $B$ are both binary matrices, as needed for the reduction to the $\norm{AB}_\infty$ problem. 

The key is to observe that 
\begin{equation}
\label{eq:product}
   A \cdot B=
  \left[ {\begin{array}{cc}
   A' + B' & \mathbf{0} \\
   \mathbf{0} & \mathbf{0} \\
  \end{array} } \right].
\end{equation}
We thus have $\norm{A \cdot B}_\infty = \norm{A' + B'}_\infty$, which is $2$ if $x \cap y \neq \emptyset$, and $1$ otherwise.  The claimed lower bound for approximating $\norm{C}_\infty$ within a factor of $2$ follows from the $\Omega(n^2)$ lower bounds for two-player set-disjointness on strings of length $\Theta(n^2)$ for success probability $0.51$ (Lemma~\ref{lem:disj}).
\end{proof}

\subsubsection{A Lower Bound for General $\kappa$-Approximation}
\label{sec:lb-infty-2}

\begin{theorem}
\label{thm:lb-infty-2}
For any $\kappa \in [1, n]$, any randomized algorithm that approximates $\norm{AB}_\infty$ for two Boolean matrices $A, B \in \{0,1\}^{n \times n}$ within a factor of $\kappa$ with probability $0.52$ needs $\tilde{\Omega}\left(n^{\frac{3}{2}} \left/ \kappa \right.\right)$ bits of communication, even if we allow an unbounded number of communication rounds.
\end{theorem}

The proof is again by a reduction from a communication problem which is highly structured.  We first introduce a few simple communication problems which will be used as building blocks to construct the final communication problem that we will use for the reduction.

Set $\beta =  \sqrt{50 \log n / n}$, and set $k = 1/(4 \kappa \beta^2)$ where $\kappa$ is the approximation ratio.
\medskip

\noindent{\em The \AND\ Problem.}
In this problem Alice holds a bit $x$ and Bob holds a bit $y$. They want to compute $\text{\AND}(x, y) = x \wedge y$.  

Let $X$ be Alice's input and $Y$ be Bob's input. We define two input distributions for $(X, Y)$.  Let $W$ be a random bit such that $\Pr[W = 0] = \Pr[W = 1] = 1/2$; let $\lambda$ be the distribution of $W$.

\begin{enumerate}
\item[$\nu_1$:] We first choose $W \sim \lambda$. If $W = 0$, we set $(X, Y) = (0, 0)$ with probability $1 - \beta$, and $(X, Y) = (0, 1)$ with probability $\beta$.  If $W = 1$, we set $(X, Y) = (0, 0)$ with probability $1 - \beta$, and $(X, Y) = (1, 0)$ with probability $\beta$.  

\item[$\mu_1$:] Set $(X, Y) = (0, 0)$ with probability $1/2$, and $(X, Y) = (1, 1)$ with probability $1/2$.
\end{enumerate} 

\medskip

\noindent{\em The \DISJ\ Problem.}  
Recall the set-disjointness problem introduced in Section~\ref{sec:preliminary}, where Alice holds $x = (x_1, \ldots, x_k) \in \{0,1\}^k$, and Bob holds $y = (y_1, \ldots, y_k) \in \{0,1\}^k$, and  they want to compute $\text{\DISJ}(x, y) = \vee_{i=1}^k \text{\AND}(x_i, y_i)$.  

Let $X = (X_1, \ldots, X_k)$ be Alice's input, and $Y = (Y_1, \ldots, Y_k)$ be Bob's input.  We again define two input distributions for $(X, Y)$.

\begin{enumerate}
\item[$\nu_k$:] Set $(X_i, Y_i) \sim \nu_1$ for each $i \in [k]$.

\item[$\mu_k$:] We first set $(X_i, Y_i) \sim \nu_k$, and then pick $M$ uniformly at random from $\{1, \ldots, k\}$, and reset $(X_M, Y_M) \sim \mu_1$.
\end{enumerate} 

\medskip

\noindent{\em The \SUM\ Problem. }  In this problem Alice holds $u = (u_1, \ldots, u_n)$ where $u_i \in \{0,1\}^k$ for each $i \in [n]$, and Bob holds $v = (v_1, \ldots, v_n)$ where $v_i \in \{0,1\}^k$ for each $i \in [n]$.  They want to compute $\text{\SUM}(u, v) = \sum_{i=1}^n \text{\DISJ}(u_i, v_i)$.  

Let $U = (U_1, \ldots, U_n)$ be Alice's input, and $V = (V_1, \ldots, V_n)$ be Bob's input. We define the following input distribution for $(U, V)$.

\begin{enumerate}
\item[$\phi$:] We first set $(U_i, V_i) \sim \nu_k$, and then pick a $D$ uniformly at random from $\{1, \ldots, n\}$, and reset $(U_D, V_D) \sim \mu_k$.
\end{enumerate} 

Note that under $(U, V) \sim \phi$,  $\Pr[\text{\SUM}(U, V) = 0] = \Pr[\text{\SUM}(U, V) = 1] = {1}/{2}$. Using the standard information complexity machinery (which we omit here; and can be found in for example \cite{WZ14,HRVZ15}) we can show the following.

\begin{theorem}
\label{thm:sum}
Any deterministic algorithm solving $\text{\SUM}(U, V)$ correctly with probability $0.51$ under $(U,V) \sim \phi$ needs $\Omega(\beta k n)$ bits of communication.
\end{theorem}

\noindent{\em Input Reduction.} We now perform a reduction from $\text{\SUM}$ to the $\ell_\infty$-norm estimation problem. Given $(U, V) \sim \phi$, we construct matrices $A$ and $B$ as follows. We set $A = [A^1, \ldots, A^{n/k}]$ where $A^1 = \ldots = A^{n/k}$, and for each $A^z\ (z \in [n/k])$ we have $A^z_{i,*} = U_i$ for all $i \in [n]$.  Similarly, we set $B = [B^1, \ldots, B^{n/k}]^T$ where $B^1 = \ldots = B^{n/k}$, and for each $B^z\ (z \in [n/k])$ we have $B^z_{*,i} = V_i$ for all $i \in [n]$.  Let $\psi$ denote the resulting distribution of $(A, B)$.  We have the following lemma.

\begin{lemma}
\label{lem:reduction}
For any $\kappa$, any deterministic algorithm that approximates $\norm{A B}_\infty$ within a factor of $\kappa$ with probability $\delta$ under $(A, B) \sim \psi$ can be used to compute $\text{\SUM}(U, V)$ with probability $(\delta + 0.01)$ under $(U, V) \sim \phi$.
\end{lemma}

\begin{proof}
Let $(U, V) \sim \phi$, and let $(A, B)$ be constructed using $(U, V)$ as described in the input reduction above.  Let $C = AB$. We first compute the value of $\norm{C}_\infty$.  

We analyze two cases.  When $\text{\SUM}(U, V) = 0$, we have $\text{\DISJ}(U_i, V_i) = 0$ for all $i \in [n]$.  Consider a pair $(i, j)\ (i, j \in [n], i \neq j)$.  We analyze the inner product $\langle A_{i,*}, B_{*,j} \rangle$.  For each $t \in [k]$, the probability that $A_{i,t} = B_{t,j} = 1$ is at most $\beta^2$. We thus have
$$\E[\langle A_{i,*}, B_{*,j} \rangle] \le \beta^2 n.
$$
By a Chernoff bound we have $\langle A_{i,*}, B_{*,j} \rangle \le 2 \beta^2 n$ with probability $1 - e^{-\beta^2 n/3} \ge 1 - 1/n^{10}$. 
By a union bound on all pairs $(i,j)\ (i \neq j)$,  we have that with probability $1 - 1/n^8$, $C_{i,j} = \langle A_{i,*}, B_{*,j} \rangle \le 2 \beta^2 n$ for all $(i,j)\ (i \neq j)$.  Consequently,
\begin{equation}
\label{eq:d-1}
\norm{C}_\infty \le 2\beta^2 n.
\end{equation}

When $\text{\SUM}(U, V) = 1$, we have $\text{\DISJ}(U_i, V_i) = 0$ for all $i \in [n] \backslash D$, and $\text{\DISJ}(U_D, V_D) = 1$.  We thus have 
\begin{equation}
\label{eq:d-2}
\norm{C}_\infty \ge n/k.
\end{equation}
By our choices of parameters $\beta$ and $k$, we have 
$$(n/k) / (2 \beta^2 n) = 2\kappa > \kappa. $$
The lemma thus follows from (\ref{eq:d-1}) and (\ref{eq:d-2}).
\end{proof}

Theorem~\ref{thm:lb-infty-2} follows from Lemma~\ref{lem:reduction}, Theorem~\ref{thm:sum}, our choices of $\beta$ and $k$, and Yao's minimax lemma.

\subsection{General Matrices}
\label{sec:non-binary}

Finally we observe that the communication complexity for approximating $\norm{AB}_\infty$ for non-binary matrices $A, B$ is significantly different than that
for binary matrices. 

\begin{theorem}
\label{thm:non-binary}

Let $A \in \mathbb{Z}^{n \times n}$ and $B \in \mathbb{Z}^{n \times n}$.  In the two-party communication model we have:
\begin{enumerate}
\item  There is an algorithm that computes $\norm{AB}_\infty$ within a factor $\kappa$ using $\tilde{O}(n^2/\kappa^2)$ bits of communication and one round.

\item Any algorithm that approximates $\norm{AB}_\infty$ within a factor $\kappa$ needs $\tilde{\Omega}(n^2/\kappa^2)$ bits of communication, even if we allow an arbitrary number of communication rounds. 
\end{enumerate}
\end{theorem}

For the upper bound, we first recall a simple algorithm for sketching $\norm{x}_\infty \ (x \in \mathbb{Z}^n)$.\footnote{This algorithm was described in \cite{SS02}.}  We first partition the vector $x$ into $n/\kappa^2$ blocks each of size $\kappa^2$, and then use the AMS sketching algorithm  \cite{AMS99} for $\ell_2$-norm estimation for each block; the sketch size is $\tilde{O}(1)$ if we target an $O(1)$-approximation and $1 - 1/n^{10}$ success probability.  Since  for each vector $y \in \mathbb{Z}^{\kappa^2}$ we have $\norm{y}_\infty \in \left[\frac{\norm{y}_2}{\kappa}, \norm{y}_2\right]$, we obtain a sketch of size $\tilde{O}(n/\kappa^2)$ for estimating $\norm{x}_\infty$ within a factor of $\kappa$.  Denote this sketching matrix by $S \in \mathbb{R}^{\tilde{O}(n/\kappa^2) \times n}$.  

In the matrix product setting Alice simply applies $S$ to $A$ and sends $SA \in \mathbb{R}^{\tilde{O}(n/\kappa^2) \times n}$ to Bob. Bob then estimates the $\ell_\infty$-norm of each column of $C (= AB)$ using $SA$ and $B$ (and computing $SA \cdot B$), and then outputs $\max_{j \in [n]}\norm{C_{*,j}}_\infty$.  

For the lower bound, we again use the technique in Section~\ref{sec:lb-infty-1} to convert a matrix product to a matrix sum, and then perform a reduction from the $\ell_\infty$-norm estimation problem (see Section~\ref{sec:preliminary}). Given two vectors $x, y \in [0, \kappa]^{n^2/4}$, we construct $A', B'$ and $A, B$ exactly the same way as that in Section~\ref{sec:lb-infty-1}.  We then have $\norm{A \cdot B}_\infty = \norm{A' + B'}_\infty$, which evaluates to $\kappa$ if $\text{\INF}(x, y) = 1$, and evaluates to at most $1$ if $\text{\INF}(x, y) = 0$. The lower bound follows from Lemma~\ref{lem:inf}.

\section{Approximate Heavy Hitters}
\label{sec:hh}

In this section we consider the $\ell_p$-$(\phi, \eps)$-heavy-hitter problem described in the introduction.  We first propose an algorithm for products of general matrices, and then consider the problem for binary matrices.

\subsection{General Matrices}
\label{sec:general-hh}
We first consider $p=1$. General $p \in (0,2]$ can be handled in a similar way.

\medskip

\noindent{\bf The Idea.\ \ }  The idea for computing approximate heavy hitters is similar to our ideas for the $\ell_\infty$-norm, that is, we sample $1$-entries in $A$ to scale down the values of entries in $C$ to a level such that the heavy-hitter entries are still non-zero, while there are not many non-zero entries corresponding to non-heavy-hitter entries.  Let $C'$ denote the matrix $C$ after we scale down. Since there cannot be many heavy hitters, the number of non-zero entries in $C'$ is small. We can thus perform a sparse recovery algorithm on $C'$ to find all the heavy hitters.  

\medskip

\noindent{\bf Algorithm.\ \ }   We present the algorithm in Algorithm~\ref{alg:hh}, and describe it in words below.

\begin{algorithm}[t]
\caption{Computing $\ell_1$-$(\phi, \eps)$-Heavy-Hitters}
\label{alg:hh}

\SetKwInOut{Input}{Input}
\SetKwInOut{Output}{Output}
\Input{Alice has a matrix $A \in \tilde{O}(n/\kappa^2)^{n \times n}$, and Bob has a matrix $B \in \tilde{O}(n/\kappa^2)^{n \times n}$. Let $C \leftarrow A B$}
\Output{$\ell_1$-$(\phi, \eps)$-Heavy-Hitters of $C$}
\BlankLine

Alice and Bob compute $\norm{C}_1$\;   \label{line:d-1}

Set the sampling rate $\beta \gets \min\left\{\frac{10^4 \log n}{\left(\frac{\eps}{\phi}\right)^2 \cdot \frac{\phi}{8} \norm{C}_1}, 1\right\}$\;  \label{line:d-2}

Alice samples each $1$-entry in $A$ with probability $\beta$ (and replaces all the non-sampled 1's by 0's), obtaining matrix $A^\beta$; let $C^\beta \gets A^\beta B$\;  

Alice and Bob then use Lemma~\ref{lem:f0} to recover all the non-zero entries of $C^\beta$; the recovered matrix $C^\beta$ is distributed at Alice's side and Bob's side, denoted by $C_A$ and $C_B$ where $C^\beta = C_A + C_B$\; \label{line:d-4}

Alice creates $C'_A$ consisting of all entries in $C_A$ that are larger than $\frac{\eps \beta}{8} \norm{C}_1$, and sends $C'_A$ to Bob.  Bob outputs all entries in $C' = C'_A + C_B$ that are at least $\beta \cdot (\phi - \frac{\eps}{2}) \norm{C}_1$.  \label{line:d-5}
\end{algorithm}

Alice and Bob first compute $\norm{C}_1$ using Remark~\ref{rem:ell-1}.  
Next, similar to Algorithm~\ref{alg:infty} for approximating $\norm{C}_\infty$, we sample the $1$-entries in matrix $A$.  The sampling is simpler in this case since we only need to sample the entries at the fixed ratio $\beta$.  Let $C^\beta$ be the resulting matrix after sampling.  

Alice and Bob then use Lemma~\ref{lem:f0} to recover all the non-zero entries in $C^\beta$; the entries of the recovered $C^\beta$ are distributed across the two parties, denoted by $C_A$ and $C_B$ where $C^\beta = C_A + C_B$.  Alice then sends all ``heavy'' entries in $C_A$, that is, those whose values are larger than $\frac{\eps \beta}{8} \norm{C}_1$, to Bob.  Bob then outputs all the heavy hitters in $C'$ which is constructed by adding the heavy entries of $C_A$ (received from Alice) to $C_B$.

\begin{theorem}
\label{thm:hh-1}
Algorithm~\ref{alg:hh} computes the $\ell_1$-$(\phi, \eps)$-heavy-hitters ($0 < \eps \le \phi \le 1$) of $A B$, where $A, B \in \mathbb{Z}^{n \times n}$, with probability $0.9$ and using $\tilde{O}(\frac{\sqrt{\phi}}{\eps} n)$ bits of communication and $O(1)$ rounds. 
\end{theorem}

We will assume that $\norm{C}_1 \ge \frac{10^4 \log n}{\left(\frac{\eps}{\phi}\right)^2 \cdot \frac{\phi}{8}} = \frac{8 \cdot 10^4 \phi  \log n }{\eps^2}$, since otherwise $\beta = 1$, and then $C^\beta = C$, in which case the  proof is only simpler. 

\medskip

\noindent{\bf Correctness.\ \ }
We define two events.  
\begin{enumerate}
\item[$\cE_6$:] For all pairs $(i,j)$, if $C_{i,j} \ge \frac{\phi}{8} \norm{C}_1$, then $C^\beta_{i,j}/\beta$ approximates $C_{i,j}$ within a factor of $1+\frac{\eps}{4\phi}$.

\item[$\cE_7$:] For all pairs $(i,j)$, if $C_{i,j} < \frac{\phi}{8} \norm{C}_1$, then $C^\beta_{i,j}/\beta < \frac{\phi}{4} \norm{C}_1$.
\end{enumerate}
The correctness of Theorem~\ref{thm:hh-1} holds if both $\cE_6$ and $\cE_7$ hold. To see this, first consider those pairs $(i,j)$ for which $C_{i,j} < \frac{\phi}{8} \norm{C}_1$.  By $\cE_7$ we have 
\begin{eqnarray*}
C'_{i,j} \le C_{i,j}^\beta
\le \beta \cdot \frac{\phi}{4} \norm{C}_1 
< \beta \cdot (\phi - \frac{\eps}{2}) \norm{C}_1.
\end{eqnarray*}
Thus pair $(i,j)$ will not be output in Step~\ref{line:d-5} of Algorithm~\ref{alg:hh}.  

We next consider those pairs $(i,j)$ with $C_{i,j} \ge \frac{\phi}{8} \norm{C}_1 $.  By $\cE_6$ we have that $C_{i,j}^\beta \in \left[\frac{\beta C_{i,j}}{1+\frac{\eps}{4\phi}}, \beta (1+\frac{\eps}{4\phi}) C_{i,j}\right]$.  Now we consider two cases.
\begin{enumerate}
\item
If $C_{i,j} \ge \phi \norm{C}_1$, then
\begin{eqnarray*}
C'_{i,j} &\ge& C^\beta_{i,j} - \frac{\eps \beta}{8} \norm{C}_1 \\ 
&\ge& \frac{\beta C_{i,j}}{1+\frac{\eps}{4\phi}} - \frac{\eps \beta}{8} \norm{C}_1 \\
&\ge& \frac{\beta \phi \norm{C}_1}{1+\frac{\eps}{4\phi}} - \frac{\eps \beta}{8} \norm{C}_1 \\
&\ge& \beta  \left(\phi - \frac{\eps}{2}\right) \norm{C}_1.
\end{eqnarray*}
Thus pair $(i,j)$ will be outputted.

\item
If $C_{i,j} < (\phi - \eps) \norm{C}_1$, then
\begin{eqnarray*}
C'_{i,j} \le \beta C^\beta_{i,j} 
&\le& \beta \left(1+\frac{\eps}{4\phi}\right) C_{i,j} \\
&<& \beta \left(1+\frac{\eps}{4\phi}\right) (\phi - \eps) \norm{C}_1 \\
&\le& \beta  \left(\phi - \frac{\eps}{2}\right)  \norm{C}_1.
\end{eqnarray*}
Thus pair $(i,j)$ will not be outputted.
\end{enumerate}

In the following we show that both $\cE_6$ and $\cE_7$ hold with probability $1 - 1/n^4$.

For $\cE_6$, for a fixed pair $(i,j)$, by sampling we have 
$$\E[C^\beta_{i,j}] = \beta \cdot C_{i,j} \ge \beta \cdot \frac{\phi}{8} \norm{C}_1. $$ 
By a Chernoff bound we have
\begin{eqnarray*}
\Pr \left[\abs{C^\beta_{i,j} - \E[C^\beta_{i,j}]}\right] &\ge& \frac{\eps}{4\phi} \cdot \E[C^\beta_{i,j}] \\ &\le&  2 \cdot e^{-(\frac{\eps}{4\phi})^2 \beta \frac{\phi}{8} \norm{C}_1 / 3} \\
&\le& 1/n^{10}.
\end{eqnarray*}
By a union bound over the at most $n^2$ $(i,j)$ pairs, we have that with probability $1 - 1/n^4$, $C^\beta_{i,j} / \beta$ approximates $C_{i,j}$ within a factor of $(1+\frac{\eps}{4\phi})$ for all pairs $(i,j)$.

For $\cE_7$, consider a fixed pair $(i,j)$.  If $C_{i,j} < \frac{\phi}{8} \norm{C}_1$, then $\E[C^\beta_{i,j}] < \beta \cdot \frac{\phi}{8} \norm{C}_1$. By a Chernoff bound we have that  $C^\beta_{i,j} \le 2 \beta \cdot \frac{\phi}{8} \norm{C}_1$ with probability $1 - 1/n^{10}$.  Thus the probability that $\cE_7$ holds is at least $1 - 1/n^4$ by a union bound over all $(i,j)$ pairs.

\medskip

\noindent{\bf Complexities.\ \ } 
Step~\ref{line:d-1} can be done using $\tilde{O}(n)$ bits (Remark~\ref{rem:ell-1}).
By a Chernoff bound, it holds with probability $1 - 1/n^{10}$ that
$\norm{C^\beta}_1 \le  2 \beta \norm{C}_1 = O\left(\frac{\phi}{\eps^2} \log n\right)$.
Consequently we have $\norm{C^\beta}_0 \le \norm{C^\beta}_1 = O(\frac{\phi}{\eps^2} \log n)$.  By Lemma~\ref{lem:f0} we have that with probability $1 - 1/n^{10}$ Alice and Bob can recover all non-zero entries of $C^\beta$ in Step~\ref{line:d-4} using $\tilde{O}(\frac{\sqrt{\phi}}{\eps} n)$ bits of communication and $2$ rounds.  
The communication in Step~\ref{line:d-5} is bounded by $\tilde{O}(1/\eps)$. We thus can bound the total communication by $\tilde{O}(\frac{\sqrt{\phi}}{\eps} n)$.  

Finally, it is easy to see that the algorithm terminates in $O(1)$ rounds.

\smallskip

The above analysis can be straightforwardly extended to $\ell_p$-norms for all constants $p \in (0, 2]$ simply by replacing the sampling probability $\beta$ by $\beta^p$ at Line~\ref{line:d-2}, and replacing $\norm{C}_1$ and matrix entries $M_{i,j}$ by $\norm{C}_p^p$ and $\abs{M_{i,j}}^p$ respectively at Lines~\ref{line:d-1}, \ref{line:d-2} and \ref{line:d-5}.  At Line~\ref{line:d-1} one can use Algorithm~\ref{alg:small-p} to estimate $\norm{C}_p^p$ up to a factor of $(1+\frac{\eps}{4\phi})$, which costs $\tilde{O}(\frac{\phi}{\eps} n)$ bits of communication by Theorem~\ref{thm:small-p}, and is a lower order term.

\begin{corollary}
For two matrices $A, B \in \mathbb{Z}^{n \times n}$, there is an algorithm that computes the $\ell_p$-$(\phi, \eps)$-heavy-hitters ($0 < \eps \le \phi \le 1, p \in (0,2]$) of $AB$ with probability $0.9$ using $\tilde{O}(\frac{\sqrt{\phi}}{\eps} n)$ bits of communication and $O(1)$ rounds. 
\end{corollary}

\subsection{Binary Matrices}
\label{sec:binary-hh}

In this section we show that we can do better for binary matrices by employing the idea we use for $\ell_\infty$-norm estimation.  Again Alice holds $A \in \{0,1\}^{n \times n}$ and Bob holds $B \in \{0,1\}^{n \times n}$, and let $C = AB$.  Due to the similarity of the approach compared with the $\ell_\infty$-norm case (Section~\ref{sec:up-infty}), we do not repeat some of the details.  

We first assume that $\norm{AB}_p^p \ge 100 \phi \log n / \eps^2$, and will consider the other case later. The algorithm is as follows.

{\em Step 1:} Alice and Bob first estimate $L_p = \norm{C}_p$ within a factor of $2$, denoted by $L'_p$. 

{\em Step 2:}  Alice samples each column of $A$ with probability $\beta = \min \left\{\frac{\alpha}{\phi^{1/p} L'_p}, 1\right\}$ for $\alpha = (10^4 \log n)^{1/p}$, obtaining $A'$.  Let $C' = A' B$.  Alice and Bob then exchange the indices of sets containing $j$ for each surviving item $j \in [n]$ as Step~\ref{line:b-2}-\ref{line:b-3} in Algorithm~\ref{alg:infty}, obtaining $C_A$ and $C_B$ for which $C' = C_A + C_B$.   

{\em Step 3:}  Alice and Bob try to verify for each non-zero entry in $C_A$ or $C_B$ whether it is indeed a heavy hitter. Let $S_A, S_B$ consist of all the entries $(i, j)$ in $C_A, C_B$ for which $(C_A)_{i,j}^p \ge \beta^p \phi (L'_p)^p/20$ or $(C_B)_{i,j}^p \ge \beta^p \phi (L'_p)^p/20$, respectively.  Then for each entry $(i,j) \in S_A \cup S_B$,  Alice and Bob try to estimate $C_{i,j}$ within a $(1 + \eps/(2\phi))$ factor by sampling $\tilde{O}(1/(\eps/\phi)^2)$ coordinates of their correponding row and column in $A$ and $B$.

By Chernoff bounds, one has that after sampling we have with probability $(1 - 1/n^{10})$ that (1) the number of sampled columns of $A$ (or, the number of surviving universe items) is bounded by $\tilde{O}(\beta n)$, and (2) $\norm{C'}_1 = \tilde{O}(\beta L_1)$.

The correctness proof is identical to that for the $\ell_\infty$-norm estimation algorithms in Section~\ref{sec:up-infty}.  We next turn to analyzing the communication cost.

The first step costs $\tilde{O}(n)$ bits of communication by Theorem~\ref{thm:small-p}.  For the second step, reusing the notation $u_j, v_j$ for each universe item $j$ in Algorithm~\ref{alg:infty}, we analyze two cases:
\begin{itemize}
\item If $\min\{u_j, v_j\} \le \sqrt{L_1/n}$, then since there are at most $\tilde{O}(\beta n)$ surviving universe items, the total communication is upper bounded by 
$$\tilde{O}(\beta n) \cdot \sqrt{\frac{L_1}{n}} = \tilde{O}\left(\frac{\sqrt{n}}{\phi^{1/p}} \cdot \frac{\sqrt{L_1}}{L_p}\right).$$

\item If $\min\{u_j, v_j\} > \sqrt{L_1/n}$, then since $\norm{C'}_1 = \tilde{O}(\beta L_1)$, the total communication is upper bounded by 

$$\tilde{O}\left(\frac{\beta L_1}{\sqrt{L_1/n}}\right) =\tilde{O}\left(\frac{\sqrt{n}}{\phi^{1/p}} \cdot \frac{\sqrt{L_1}}{L_p}\right).$$
\end{itemize}
It is easy to see that the third step costs $\tilde{O}((\phi/\eps)^2 \cdot 1/\phi) = \tilde{O}(\phi/\eps^2)$ bits of communication since there can be at most $\tilde{O}(1/\phi)$ entries whose $p$-th powers are at least $\beta^p \phi {L'}_p^p / 20$.  Summing up, the total communication is bounded by $\tilde{O}(Z)$ where
\begin{eqnarray*}
\label{eq:z-1}
Z & = &n + \frac{\sqrt{n}}{\phi^{1/p}} \cdot \frac{\sqrt{L_1}}{L_p} + \frac{\phi}{\eps^2} \\
&\le& n + \frac{\phi}{\eps^2} + \frac{n^{\frac{1}{2}}}{\phi^{1/p}} \cdot \frac{\sqrt{L_1}}{L_2 / (n^{\frac{1}{2} - \frac{1}{p}})} \\
&\le& n + \frac{\phi}{\eps^2} + \frac{n^{1 - \frac{1}{p}}}{\phi^{1/p}}  \quad (\sqrt{L_1} \le L_2) \\
&\le& 2\left(n + \frac{\phi}{\eps^2}\right). \quad \left(\frac{\phi}{\eps^2} \ge \frac{1}{\phi}\right)
\end{eqnarray*}

In the case that $\norm{AB}_p^p < 100 \phi \log n / \eps^2$, we can just omit the subsampling in Step $2$ of the algorithm.  A similar analysis gives a communication cost of $\tilde{O}(n + \frac{\sqrt{\phi n}}{\eps} + \frac{1}{\eps}) = \tilde{O}(n + \frac{\phi}{\eps^2})$.

\begin{theorem}
There is an algorithm that computes the $\ell_p$-$(\phi, \eps)$-heavy-hitters ($0 < \eps \le \phi \le 1, p \in (0, 2]$) of $A B$, where $A, B \in \{0,1\}^{n \times n}$, with probability $0.9$ and using $\tilde{O}(n + \frac{\phi}{\eps^2})$ bits of communication and $O(1)$ rounds.
\end{theorem}

\section{Concluding Remarks}
\label{sec:conclude}
In this paper we studied a set of basic statistical estimation problems of matrix products in the distributed model, including the $\ell_p$-norms, distinct elements, $\ell_0$-sampling and heavy hitters.  These problems have a number of applications in database joins.

We would like to mention again that our algorithms for square matrices can be straightforwardly modified to handle rectangular matrices where $A \in \Sigma^{m \times n}\ (m \ge n)$ and $B \in \Sigma^{n \times m}$.  We briefly list here how our main upper bounds look like on rectangular matrices.  All the algorithms remain the same (we of course have to change some occurrences of $n$ to $m$ in several places).
\begin{itemize}
\item
The communication cost for $(1+\eps)$-approximating $\ell_p$ $(p \in [0,2])$ with $\Sigma = \mathbb{Z}$ remains $\tilde{O}(n/\eps)$. 
\item
The communication cost for $(2+\eps)$-approximating $\ell_\infty$ with $\Sigma = \{0,1\}$ becomes $\tilde{O}(m^{1.5})$, and that for $\kappa$-approximating $\ell_\infty$ with $\Sigma = \{0,1\}$  becomes $\tilde{O}(m^{1.5}/\kappa)$
\item
The communication cost for $\ell_p$-$(\phi, \eps)$-heavy-hitters with $\Sigma = \mathbb{Z}$ remains $\tilde{O}(\frac{\sqrt{\phi}}{\eps} n)$, and that for $\ell_p$-$(\phi, \eps)$-heavy-hitters with $\Sigma = \{0,1\}$ remains $\tilde{O}(n + \frac{\phi}{\eps^2})$.
\end{itemize}

\bibliographystyle{abbrv}
\balance
\bibliography{paper}

\end{document}